\documentclass[12pt]{article}

\usepackage[utf8]{inputenc}
\usepackage{amsmath}
\usepackage{amsthm}
\usepackage{amssymb}
\usepackage{csquotes}
\usepackage{verbatim}
\usepackage{tikz}
\usepackage{mathtools}
\usepackage{algorithm}
\usepackage{algorithmicx}
\usepackage[noend]{algpseudocode}
\usepackage{graphicx}
\usepackage[export]{adjustbox}
\usepackage{todonotes}
\usepackage{hyperref}

\usepackage{tikz}
\usetikzlibrary{backgrounds}
\usetikzlibrary{arrows}
\usetikzlibrary{shapes,shapes.geometric,shapes.misc}
\tikzstyle{tikzfig}=[baseline=-0.25em,scale=0.5]
\pgfkeys{/tikz/tikzit fill/.initial=0}
\pgfkeys{/tikz/tikzit draw/.initial=0}
\pgfkeys{/tikz/tikzit shape/.initial=0}
\pgfkeys{/tikz/tikzit category/.initial=0}
\pgfdeclarelayer{edgelayer}
\pgfdeclarelayer{nodelayer}
\pgfsetlayers{background,edgelayer,nodelayer,main}
\tikzstyle{every loop}=[]
\newcommand{\inlinetikz}[3]{\raisebox{#1}{\scalebox{#2}{#3}}}

\tikzset{%
    -,
    auto,
    thick,
    every node/.style={line width=1pt },
    main node/.style={draw,circle,fill=white},
    gray node/.style={draw,circle,fill=gray},
    main edge/.style={line width=2pt},
    orange edge/.style={draw=orange, line width=3pt},
    secondary edge/.style={draw=gray, line width=1pt},
}

\def\extriangle{
    \begin{tikzpicture}
        \begin{pgfonlayer}{nodelayer}
            \node [style=main node] (4) at (0, 0) {};
            \node [style=main node] (5) at (0, 1) {};
            \node [style=gray node] (6) at (0.5, 0.5) {};
        \end{pgfonlayer}
        \begin{pgfonlayer}{edgelayer}
            \draw [style=main edge] (4) to (5);
            \draw [style=main edge] (5) to (6);
            \draw [style=main edge] (6) to (4);
        \end{pgfonlayer}
    \end{tikzpicture}
}

\def\exsquare{
    \begin{tikzpicture}
        \begin{pgfonlayer}{nodelayer}
            \node [style=gray node] (0) at (0.5, 0.5) {};
            \node [style=main node] (1) at (1, 0) {};
            \node [style=main node] (2) at (1.5, 0.5) {};
            \node [style=main node] (3) at (1, 1) {};
        \end{pgfonlayer}
        \begin{pgfonlayer}{edgelayer}
            \draw [style=main edge] (0) to (1);
            \draw [style=main edge] (1) to (2);
            \draw [style=main edge] (2) to (3);
            \draw [style=main edge] (3) to (0);
        \end{pgfonlayer}
    \end{tikzpicture}
}

\def\exfish{
    \begin{tikzpicture}
        \begin{pgfonlayer}{nodelayer}
            \node [style=gray node] (0) at (0.5, 0.5) {};
            \node [style=main node] (1) at (1, 0) {};
            \node [style=main node] (2) at (1.5, 0.5) {};
            \node [style=main node] (3) at (1, 1) {};
            \node [style=main node] (4) at (0, 0) {};
            \node [style=main node] (5) at (0, 1) {};
            \node [style=gray node] (6) at (0.5, 0.5) {};
        \end{pgfonlayer}
        \begin{pgfonlayer}{edgelayer}
            \draw [style=main edge] (0) to (1);
            \draw [style=main edge] (1) to (2);
            \draw [style=main edge] (2) to (3);
            \draw [style=main edge] (3) to (0);
            \draw [style=main edge] (4) to (5);
            \draw [style=main edge] (5) to (6);
            \draw [style=main edge] (6) to (4);
        \end{pgfonlayer}
    \end{tikzpicture}
}

\def\exextension{
    \begin{tikzpicture}
        \begin{pgfonlayer}{nodelayer}
            \node [style=main node] (0) at (0, 0) {};
            \node [style=main node] (1) at (1, 0) {};
            \node [style=main node] (2) at (2.5, 0) {};
            \node [style=main node] (3) at (3.5, 0) {};
            \node [style=main node] (4) at (4.5, 0) {};
            \node [style=main node] (5) at (6, -1) {};
            \node [style=main node] (6) at (7, -1) {};
            \node [style=main node] (7) at (6.5, 0) {};
            \node [style=main node] (8) at (6.5, -2) {};
            \node [style=main node] (9) at (4.5, -2) {};
            \node [style=main node] (10) at (3.5, -2) {};
            \node [style=main node] (11) at (2.5, -2) {};
            \node [style=main node] (12) at (1, -2) {};
            \node [style=main node] (13) at (0, -2) {};
        \end{pgfonlayer}
        \begin{pgfonlayer}{edgelayer}
            \draw [style=main edge] (0) to (1);
            \draw [style=main edge] (2) to (3);
            \draw [style=main edge] (3) to (4);
            \draw [style=main edge] (6) to (7);
            \draw [style=main edge] (7) to (5);
            \draw [style=orange edge] (5) to (6);
            \draw [style=main edge] (13) to (12);
            \draw [style=main edge] (11) to (10);
            \draw [style=main edge] (10) to (9);
            \draw [style=main edge] (8) to (5);
            \draw [style=main edge] (8) to (6);
            \draw [style=secondary edge] (9) to (7);
            \draw [style=secondary edge] (10) to (7);
            \draw [style=secondary edge] (11) to (7);
            \draw [style=secondary edge] (12) to (7);
            \draw [style=secondary edge] (13) to (7);
            \draw [style=secondary edge] (9) to (4);
            \draw [style=secondary edge] (9) to (3);
            \draw [style=secondary edge] (9) to (2);
            \draw [style=secondary edge] (9) to (1);
            \draw [style=secondary edge] (9) to (0);
            \draw [style=secondary edge] (10) to (4);
            \draw [style=secondary edge] (10) to (3);
            \draw [style=secondary edge] (10) to (2);
            \draw [style=secondary edge] (10) to (1);
            \draw [style=secondary edge] (10) to (0);
            \draw [style=secondary edge] (11) to (4);
            \draw [style=secondary edge] (11) to (3);
            \draw [style=secondary edge] (11) to (2);
            \draw [style=secondary edge] (11) to (1);
            \draw [style=secondary edge] (11) to (0);
            \draw [style=secondary edge] (12) to (4);
            \draw [style=secondary edge] (12) to (3);
            \draw [style=secondary edge] (12) to (2);
            \draw [style=secondary edge] (12) to (1);
            \draw [style=secondary edge] (12) to (0);
            \draw [style=secondary edge] (13) to (4);
            \draw [style=secondary edge] (13) to (3);
            \draw [style=secondary edge] (13) to (2);
            \draw [style=secondary edge] (13) to (1);
            \draw [style=secondary edge] (13) to (0);
            \draw [style=secondary edge] (4) to (8);
            \draw [style=secondary edge] (3) to (8);
            \draw [style=secondary edge] (2) to (8);
            \draw [style=secondary edge] (1) to (8);
            \draw [style=secondary edge] (0) to (8);
            \draw [style=secondary edge, bend left=90, looseness=2.00] (7) to (8);
        \end{pgfonlayer}
    \end{tikzpicture}
}

\def\exgadgetone{
    \begin{tikzpicture}
        \begin{pgfonlayer}{nodelayer}
            \node [style=main node] (0) at (0, 0) {};
            \node [style=main node] (1) at (1, 0) {};
            \node [style=gray node] (2) at (2, -0.5) {};
            \node [style=main node] (3) at (1, -1) {};
            \node [style=main node] (4) at (0, -1) {};
        \end{pgfonlayer}
        \begin{pgfonlayer}{edgelayer}
            \draw [style=main edge] (1) to (2);
            \draw [style=main edge] (2) to (3);
            \draw [style=secondary edge] (0) to (4);
            \draw [style=secondary edge] (0) to (3);
            \draw [style=secondary edge] (1) to (3);
            \draw [style=secondary edge] (1) to (4);
        \end{pgfonlayer}
    \end{tikzpicture}
}

\def\exgadgettwo{
    \begin{tikzpicture}
        \begin{pgfonlayer}{nodelayer}
            \node [style=main node] (1) at (1, 0) {};
            \node [style=main node] (2) at (2, 0) {};
            \node [style=gray node] (3) at (0, -0.5) {};
            \node [style=main node] (4) at (1, -1) {};
            \node [style=main node] (5) at (2, -1) {};
        \end{pgfonlayer}
        \begin{pgfonlayer}{edgelayer}
            \draw [style=main edge] (2) to (1);
            \draw [style=main edge] (5) to (4);
            \draw [style=secondary edge] (1) to (4);
            \draw [style=secondary edge] (1) to (5);
            \draw [style=secondary edge] (4) to (2);
            \draw [style=secondary edge] (2) to (5);
        \end{pgfonlayer}
    \end{tikzpicture}
}

\def\exgadgetthree{
    \begin{tikzpicture}
        \begin{pgfonlayer}{nodelayer}
            \node [style=main node] (2) at (2, 0) {};
            \node [style=main node] (5) at (2, -1) {};
            \node [style=gray node] (6) at (1, 0) {};
            \node [style=main node] (7) at (0, 0) {};
            \node [style=main node] (8) at (0, -1) {};
        \end{pgfonlayer}
        \begin{pgfonlayer}{edgelayer}
            \draw [style=main edge] (6) to (2);
            \draw [style=main edge] (6) to (5);
            \draw [style=secondary edge] (2) to (5);
            \draw [style=secondary edge] (8) to (2);
            \draw [style=secondary edge] (7) to (8);
            \draw [style=secondary edge] (7) to (5);
        \end{pgfonlayer}
    \end{tikzpicture}
}

\def\exgraphone{
    \begin{tikzpicture}
        \begin{pgfonlayer}{nodelayer}
            \node [style=main node] (0) at (0, 0) {};
            \node [style=main node] (1) at (1, 0) {};
            \node [style=main node] (2) at (2, 0) {};
        \end{pgfonlayer}
        \begin{pgfonlayer}{edgelayer}
            \draw [style=main edge] (2) to (1);
        \end{pgfonlayer}
    \end{tikzpicture}
}

\def\exgraphtwo{
    \begin{tikzpicture}
        \begin{pgfonlayer}{nodelayer}
            \node [style=main node] (0) at (0, 0) {};
            \node [style=main node] (1) at (1, 0) {};
        \end{pgfonlayer}
        \begin{pgfonlayer}{edgelayer}
            \draw [style=main edge] (0) to (1);
        \end{pgfonlayer}
    \end{tikzpicture}
}

\def\opt{{\it opt}}

\def\F{{\cal F}}
\def\G{{\cal G}}
\def\N{{\bf N}}

\def\isub{\trianglelefteq} 
\def\phisub{\isub_\phi}
\let\phi=\varphi
\def\join{\nabla}
\def\iso{\simeq}
\newcommand{\set}[2]{\ensuremath{\left\{\, #1 \mid #2 \,\right\}}}
\newcommand{\NDD}[1]{\textsc{Delayed $#1$-Node-Deletion Problem}}
\def\FNDD{\NDD{\F}}
\def\HNDD{\NDD{H}}
\newcommand{\EDD}[1]{\textsc{Delayed $#1$-Edge-Deletion Problem}}
\def\FEDD{\EDD{\F}}
\def\HEDD{\EDD{H}}

\theoremstyle{definition} % Define theorem styles here based on the definition style (used for definitions and examples)
\newtheorem{definition}{Definition}
\theoremstyle{plain} % Define theorem styles here based on the plain style (used for theorems, lemmas, propositions)
\newtheorem{theorem}{Theorem}
\newtheorem{lemma}{Lemma}
\newtheorem{corollary}{Corollary}

\begin{document}
\title{Advice Complexity bounds for Online Delayed $\F$-Node-, $H$-Node- and $H$-Edge-Deletion Problems}
\author{
Niklas Berndt\and
Henri Lotze}
\date{RWTH Aachen University, Germany}
%\authorrunning{N.~Berndt, H.~Lotze and P.~Rossmanith}
%\titlerunning{Advice Complexity for the Online Delayed $H$-Node and $H$-Edge Deletion Problems}
%\institute{Department of Computer Science, RWTH Aachen University, Ahornstr. 55, 52074 Aachen, Germany\\ 
%\email{\{burjons,gehnen,lotze,mock,rossmani\}@cs.rwth-aachen.de}
%}
\maketitle
\begin{abstract}
    Let $\F$ be a fixed finite obstruction set of graphs and
    $G$ be a graph revealed in an online fashion, node by node.  The online {\sc Delayed
    $\F$-Node-Deletion Problem} ({\sc
    $\F$-Edge-Deletion Problem}) is to keep $G$ free of every $H \in \F$
    by deleting nodes (edges) until no induced subgraph isomorphic to any graph
    in~$\F$ can be found in $G$. The task is to keep the number of
    deletions minimal.
    
    Advice complexity is a model in which an online algorithm has
    access to a binary tape of infinite length, on which an oracle can
    encode information to increase the performance of the algorithm.
    We are interested in the minimum number of advice bits that
    are necessary and sufficient to solve a deletion problem optimally.
    
    In this work, we first give essentially tight bounds on the advice
    complexity of the {\sc Delayed
    $\F$-Node-Deletion Problem} and {\sc
    $\F$-Edge-Deletion Problem} where $\F$ consists of a single, arbitrary graph $H$.
    We then show that the gadget used to prove
    these results can be utilized to give tight bounds in the case of node deletions if $\F$ consists of
    either only disconnected graphs or only connected graphs. Finally, we show that the number
    of advice bits that is necessary and sufficient to solve the general {\sc Delayed $\F$-Node-Deletion Problem} is heavily dependent
    on the obstruction set~$\F$. To this end, we provide sets for which
    this number is either constant, logarithmic or linear in the
    optimal number of deletions.
\end{abstract}

\section{Introduction}
The analysis of online problems is concerned with studying the worst case
performance of algorithms where the instance is revealed element by element and decisions 
have to be made immediately and irrevocably. To measure the performance of such
algorithms, their solution is compared to the optimal solution of the same
instance. The largest ratio between the size of an online algorithms solution
and the optimal solution size over all instances is then called the (strict) \emph{competitive
ratio} of an algorithm. Finding an algorithm with the smallest possible
competitive ratio is the common aim of online analysis. The study of online
algorithms was started by Sleator and Tarjan~\cite{SleatorT84} and has been
active ever since. For a more thorough introduction on competitive analysis, we refer
the reader to the standard book by Borodin and El-Yaniv~\cite{Borodin1998}.

The online problems studied in this work are each defined over a fixed family
$\F$ of graphs. An induced online graph $G$ is revealed iteratively by revealing its
nodes. The solution set, i.e.\ a set of nodes (edges), of an algorithm
is called $S$ and we define $G - S$ as $V(G) \setminus V(S)$ or as $E(G) \setminus E(S)$
respectively, depending on whether $S$ is a set of nodes or edges. When in some
step $i$ an induced subgraph of $G[\{v_1,\ldots,v_i\}] - S$ is isomorphic to a
graph $H \in \F$, an algorithm is forced to delete nodes (edges) $T$ by adding
them to $S$ until no induced graph isomorphic to some $H \in \F$ can be found in
$G[\{v_1,\ldots,v_i\}] - \{S \cup T\}$. The competitive ratio of an algorithm is then
measured by taking the ratio of its solution set size to the solution set size
of an optimal offline algorithm.

Note that this problem definition is not compatible with the classical online
model, as nodes (edges) do not immediately have to be added to $S$ or ultimately
not be added to $S$. Specifically, elements that are not yet part of $S$ may be
added to $S$ at a later point, but no elements may be removed from $S$ at any
point. Furthermore, an algorithm is only forced to add elements to $S$ whenever
an $H \in \F$ is isomorphic to some induced subgraph of the current online
graph. Chen et al.~\cite{ChenHLR21} showed that no algorithm for this problem can admit
a constantly bounded competitive ratio in the classical online setting and that there are
families $\F$ for which the competitive ratio is strict in the size of the
largest forbidden graph $H \in \F$.
This model, where only an incremental valid partial solution is to be upheld,
was first studied by Boyar et al.~\cite{Boyar16} and coined ``Late Accept''
by Boyar et al.~\cite{Boyar17} in the following year.
As we study the same problems as Chen et al., we use the term
\emph{delayed} for consistency.

When studying the competitive ratio of online algorithms, Dobrev et
al.~\cite{DobrevKP09} asked the question which, and crucially \emph{how much}
information an online algorithm is missing in order to improve its competitive
ratio. This model was revised by Hromkovi\v{c} et al.~\cite{HromkovicKK10},
further refined by B\"{o}ckenhauer et al.~\cite{BockenhauerKKKM17} and is known
as the study of \emph{advice complexity} of an online problem. In this setting,
an online algorithm is given access to a binary advice tape that is infinite in
one direction and initialized with random bits. An oracle may then overwrite a
number of these random bits, starting from the initial position of the tape in
order to encode information about the upcoming instance or to simply give
instructions to an algorithm. An algorithm may then read from this tape and act
on this information during its run. The maximum number of bits an algorithm
reads from the tape over all instances to obtain a target competitive ratio of $c$
is then called the \emph{advice complexity} of
an algorithm. For a more thorough introduction to the analysis of advice
complexity and problems studied under this model, we refer the reader
to the book by Komm~\cite{Komm16} and the
survey paper by Boyar et al.~\cite{BoyarFKLM17}. In this
work, we are interested in the minimum needed and maximum necessary number of
bits of information an online algorithm needs in order to solve the discussed
problems optimally.

The analysis of advice complexity assumes the existence of an
almighty oracle that can give perfect advice, which is not realistic.
However, recent publications utilize such bounds,
especially the information theoretic lower bounds. One field is that
of machine-learned advice~\cite{AhmadianEMP22,LykourisV21,LindermayrMS22}. A related field is that of
uncertain or untrusted advice~\cite{LindermayrM22}, which analyses the
performance of algorithms depending on
how accurate some externally provided information for an algorithm is.

%We mention in a previous paragraph already that the complexity is
%measured in the size of the solution set, this paragraph may be
%redundant
%A noticeable consequence of analyzing \emph{delayed} online problems is that it
%is not always reasonable to analyze the competitive ratio --- or advice complexity --- in
%the size of the input, as the instance size is not necessarily in any relation
%to the solution size. A simple example is $\F = \{v_1v_2\}$, i.e. an obstruction
%set consisting of an edge between two nodes and an online instance that consists of one
%edge and an arbitrary number of isolated nodes. We thus measure the complexity in the
%size of the solution set.

Chen et al.~\cite{ChenHLR21} gave bounds on the \FNDD{} and \FEDD{}
for several restricted families $\F$, dividing the problem by
restricting whether $\F$ may contain only connected graphs and whether
$\F$ consists of a single graph $H$ or arbitrarily many. Tables
\ref{tab:node} and~\ref{tab:edge} show our contributions: Up to some
minor gap due to encoding details, we close the remaining gap for the
\HNDD{} that was left open by Chen et al.~\cite{ChenHLR21}, give tight
bounds for the \HEDD, and provide a variety of (tight) bounds for the
\FNDD{} depending on the nature of $\F$.

{\renewcommand{\arraystretch}{1.1}%
\begin{table}[t]
    \caption{Advice complexity of node-deletion problems.}
    \label{tab:node}
    \centering
    \begin{tabular}{l @{\hskip 2.5mm}|@{\hskip 2.5mm} l @{\hskip 5mm}l}
        Node-Deletion & Single graph $H$ forbidden & Family $\F$ of graphs forbidden\\
        \hline
        All graphs connected & Chen et al.~\cite{ChenHLR21}& Chen et al.~\cite{ChenHLR21}\\

        Arbitrary graphs & Essentially tight bound & Lower \& Upper bounds\\
                         & {\hfill \scriptsize Thm.~\ref{thm:trivial_ub_node_deletion}, Cor.~\ref{cor:hndd}} & {\hfill \scriptsize Thm.~\ref{thm:trivial_ub_node_deletion}, Thm.~\ref{thm:nodelogkub}, Thm.~\ref{thm:nodelogklbuni}, Thm.~\ref{thm:nodelogklbisol}}
    \end{tabular}
\end{table}
}
{\renewcommand{\arraystretch}{1.1}%
\begin{table}[t]
    \caption{Advice complexity of edge-deletion problems.}
    \label{tab:edge}
    \centering
    \begin{tabular}{l @{\hskip 2.5mm}|@{\hskip 2.5mm} l @{\hskip 5mm} l}
        Edge-Deletion & Single graph $H$ forbidden & Family $\F$ of graphs forbidden\\
        \hline
        All graphs connected & Chen et al.~\cite{ChenHLR21} & Chen et al.~\cite{ChenHLR21} \\

        Arbitrary graphs & Essentially tight bound & Open\\
                         & {\hfill \scriptsize Thm.~\ref{thm:trivial_ub_edge_deletion}, Cor.~\ref{cor:heddnoisol}, Thm.~\ref{thm:heddwithisol}}
    \end{tabular}
\end{table}
}

The problem for families of \emph{connected} graphs has one nice characteristic
that one can exploit when designing algorithms solving node- or
edge-deletion problems over such families. Intuitively, one can build
an adversarial instance by presenting some copy of an $H \in \F$,
extend this $H$ depending on the behavior of an algorithm to force
some specific deletion and continue by presenting the next copy of
$H$. The analysis can then focus on each individual copy of $H$ to
determine the advice complexity. For families of \emph{disconnected}
graphs, this is not always so simple: Remnants of a previous copy of
some $H$ together with parts of another copy of $H$ may themselves be a
copy of some other $H' \in \F$. Thus, while constructing gadgets for
some singular copy of $H \in \F$ is usually simple and can force an
algorithm to distinguish the whole edge set of $H$, this generally
breaks down once the next copy of $H$ is presented in the same
instance.

The rest of this work is structured as follows. We first give formal definitions
of the problems that we analyze in this work and introduce necessary notation. We
then give tight
bounds on the \FNDD{} for completely connected and completely disconnected
families $\F$ and analyze the exact advice complexity of the \HEDD{} and \HNDD{}.
We then take a closer look at the general \FNDD{}, showing that its
advice complexity is heavily dependent on the concrete obstruction set $\F$.
To this end, we show that depending on $\F$, constant, logarithmic or linear advice
can be necessary and sufficient to optimally solve the \FNDD.

%Reference: possibly advice preserving reductions, past publications about the
%problem.
\subsection{Notation and Problem Definitions}
We use standard graph notation in this work. Given an undirected graph $G =
(V,E)$, $G[V']$ with $V'\subseteq V$ denotes the graph induced by the node set
$V'$. We use $|G|$ to denote $|V(G)|$ and $||G||$ to denote $|E(G)|$. We use the
notation $\overline{G}$ for the complement graph of $G$. $K_n$ and
$\overline{K_n}$ denote the $n$-clique and the $n$-independent set respectively. $P_n$ denotes the path on $n$ nodes.
The neighborhood of a vertex $v$ in a graph $G$ consists of all vertices adjacent
to $v$ and is denoted by $N^G(v)$. A vertex $v$ is called universal, if $N^G(v) = V \setminus \{v\}$.
We measure the advice complexity in $opt$, which denotes the size of the optimal
solution of a given problem.

We adapt some of the notation of Chen et al.~\cite{ChenHLR21}. We use
$H \phisub G$ for graphs $H$ and $G$ iff there exists an
isomorphism~$\phi$ such that $\phi(H)$ is an induced subgraph of $G$.
A graph $G$ is called \emph{$\F$-free} if there is no $H \phisub G$ for
any $H \in \F$. Furthermore, a \emph{gluing} operation works as
follows: Given two graphs $G$ and $G'$, identify a single node from
$G$ and a single node from $G'$. For example, if we glue
\inlinetikz{-3pt}{0.3}{\extriangle}
together with \inlinetikz{-3pt}{0.3}{\exsquare} at the gray nodes, the resulting graph is
\inlinetikz{-3pt}{0.3}{\exfish}.

\begin{definition}
    Let $\F$ be a fixed family of graphs. Given an online graph $G$ induced by
    its nodes $V(G) = \{v_1,\ldots,v_n\}$, ordered by their occurrence
    in an online instance. The $\FNDD$ is for every $i$ to select a set $S_i
    \subseteq \{v_1,\ldots,v_i\}$ such that $G[\{v_1,\ldots,v_i\}] - S_i$ is $\F$-free.
    Furthermore, it has to hold that $S_1 \subseteq\ldots\subseteq S_n$,
    where $|S_n|$ is to be minimized.
\end{definition}
The definition of the $\FEDD$ is identical, with $S$ being a set of edges of
$G$ instead of nodes. If $\F = \{H\}$ we speak of an {\sc H-Deletion Problem} instead of an {\sc\{H\}-Deletion Problem}.

For an obstruction set $\F$ we assume that there exist no distinct
$H_1, H_2 \in \F $ with $H_1 \phisub H_2$ for some isomorphism $\phi$,
as each online graph containing $H_2$ also contains $H_1$, making
$H_2$ redundant.  Furthermore, we assume in the case of the $\FEDD$
that $\F$ contains no $\overline{K_n}$ for any $n$. This assumption is
arguably reasonable as no algorithm that can only delete edges is able
to remove a set of isolated nodes from a graph.

\section{Essentially Tight Advice Bounds}
One can easily construct a naive online algorithm for the \FNDD{} that is provided a complete optimal solution on the advice tape and deletes nodes accordingly. The online algorithm does not make any decisions itself, but strictly follows the solution provided by the oracle. The resulting trivial upper bounds on the advice complexity of the problem is summarized in the following theorem.
%If the assumption from the notation stays that no IS should be part of an \F in Edge-Deletion, we can simplify the notation
\begin{theorem}\label{thm:trivial_ub_node_deletion}
    Let $\F$ be an arbitrary family of graphs, and let $H \in \F$ be a maximum order graph. Then there is an optimal online algorithm for the \FNDD{} with advice that reads at most $\opt \cdot \log \vert H \vert + O(\log\opt)$ bits of advice.
\end{theorem}
\begin{proof}
    The online algorithm reads $opt$ from the tape using self-delimiting encoding (see \cite{Komm16}). Then it reads $\lceil \opt \cdot \log \vert H \vert \rceil$ bits and interprets them as $\opt$ numbers $d_1,...,d_\opt \leq \vert H \vert$. Whenever some forbidden induced graph is detected, the algorithm deletes the $d_i$th node.
    
\end{proof}

With the same idea we can easily get a similar upper bound for the \HEDD{}.
\begin{theorem}\label{thm:trivial_ub_edge_deletion}
    There is an optimal online algorithm with advice for the \HEDD{} that reads at most $\opt \cdot \log \| H \| + O(\log \opt)$ bits of advice.
\end{theorem}

In this section we show that for the \FNDD{} for connected or disconnected $\F$, as well as for the \HEDD{} this naive strategy is already the best possible. More formally, we meet these trivial upper bounds by essentially tight lower bounds for the aforementioned problems. We call lower bounds of the form $\opt \cdot \log \vert H \vert$, or $\opt \cdot \log \| H \|$ \emph{essentially tight}, because they only differ from the trivial upper bound by some logarithmic term in $\opt$. This additional term stems from the fact that the advisor must encode $\opt$ onto the advice tape in order for the online algorithm to correctly interpret the advice. If the online algorithm knew $\opt$ in advance, we would have exactly tight bounds. 

\subsection{Connected and Disconnected $\F$-Node Deletion Problems}
Chen et al.~\cite{ChenHLR21} previously proved essentially tight bounds on the advice complexity of the \FNDD{} for the case that all graphs in $\F$ are connected. They found 
a lower bound of $\opt \cdot \log \vert H \vert$ where $H$ is a maximum order graph in $\F$. Additionally, they proved a lower bound on the advice complexity of the \HNDD{} for disconnected $H$ that was dependent on a maximum order connected component $C_{max}$ of $H$: $\opt \cdot \log \vert C_{max} \vert + O(\log \opt)$. We improve this result and provide a lower bound on the advice complexity of the \FNDD{} for families $\F$ of disconnected graphs, that essentially matches the trivial upper bound from Theorem \ref{thm:trivial_ub_node_deletion}.

\begin{lemma}\label{thm:duality}
    Let $\F$ be an arbitrary obstruction set, and let $\overline{\F} := \set{\overline{H}}{H \in \F}$ be the family of complement graphs. Then the advice complexity of the \FNDD{} is the same as for the \NDD{\overline{\F}}.
\end{lemma}

\begin{proof}
    We provide an advice preserving reduction from \NDD{\F} to \NDD{\overline{\F}}: If $G$ is an online instance for the $\F$-problem, such that in the $i$th time step $G[\{v_1,...,v_i\}]$ is revealed, then we construct $\overline{G}$ as an online instance for the $\overline{\F}$-problem by revealing $\overline{G[\{v_1,...,v_i\}]}$ in the $i$th time step. An optimal online algorithm has to delete the same nodes in the same time steps in both instances. This proves that the advice complexity for the $\F$-problem is at most the advice complexity for the $\overline{\F}$-problem. The same reduction in the other direction yields equality.
    
\end{proof}

From this follows immediately the desired lower bound on the advice complexity of the \FNDD{} for disconnected $\F$.

\begin{theorem}\label{thm:lb_disconnected}
    Let $\F$ be a family of disconnected graphs, and $H \in \F$ a maximum order graph. Then any optimal online algorithm for the \FNDD{} needs $\opt \cdot \log \vert H \vert$ bits of advice.
\end{theorem}
\begin{proof}
    Since all graphs in $\F$ are disconnected, $\overline{\F}$ is a family of connected graphs. For $\overline{\F}$, the lower bound proven by Chen et al. of $\opt \cdot \log \vert H \vert$ holds. The claim follows from Lemma \ref{thm:duality}.
\end{proof}

In summary, this lower bound of $\opt \cdot \log \left( \max_{H \in \F} \vert H \vert \right)$ holds for all families of graphs $\F$ that contain either only connected or only disconnected graphs. In particular, these results imply a tight lower bound on the advice complexity of the \HNDD~ for arbitrary graphs $H$, which Chen et al. raised as an open question.

\begin{corollary}\label{cor:hndd}
    Let $H$ be an arbitrary graph. Then any online algorithm for the \HNDD~ requires $\opt \cdot \log \vert H \vert$ bits of advice to be optimal.
\end{corollary}

We want to briefly reiterate the main steps of the lower bound proof by Chen et al. for connected $\F$. Let $H$ be a maximum order graph in $\F$. The idea is to construct $\vert H \vert^{\opt}$ different instances with optimal solution size $\opt$ such that no two of these instances can be handled optimally with the same advice string. These instances consist of the disjoint unions of $\opt$ gadgets where each gadget is constructed by gluing two copies of $H$ at an arbitrary node. This way, each gadget needs at least one node deletion, and deleting the glued node is the only optimal way to make a gadget $\F$-free. Since in each of the $\opt$ gadget we have $\vert H \vert$ choices of where to glue the two copies together, we in total construct $\vert H \vert^{\opt}$ instances. This procedure of constructing instances that have to be handled by different advice strings is a standard method of proving lower bounds on the advice complexity.

Since the proof of Theorem \ref{thm:lb_disconnected} uses this result by Chen et al. one can examine closer the instances that are constructed implicitly in this proof.
As a kind of ``dual'' approach to the disjoint gadget constructions of Chen et al., we will use the so-called \emph{join} of graphs.
%It will turn out to be useful to consider the so-called \emph{join} of graphs.

\begin{definition}
    Given two graphs $G_1,G_2$.
    The \emph{join graph} $G = G_1 \join G_2$ is constructed by
    connecting each vertex of $G_1$ with each vertex of $G_2$ with an
    edge, i.e.\ $V(G) = V(G_1) \cup V(G_2)$, $E(G) = E(G_1) \cup E(G_2) \cup
    \{\,v_1v_2\mid v_1 \in V(G_1), v_2 \in V(G_2)\,\}$.
\end{definition}

First of all we look at how the gadgets for disconnected $H$ look now. Let $H$ be a maximum order graph of $\F$ where $\F$ consists of disconnected graphs. Then a gadget of $H$ is the complement graph of a gadget of the connected graph $\overline{H}$. Therefore, a gadget of $H$ is constructed by gluing two copies of $H$ at some arbitrary vertex and then joining them everywhere else. For example, for the graph \inlinetikz{1pt}{0.3}{\exgraphone} has the following three possible gadgets: \inlinetikz{-3pt}{0.3}{\exgadgettwo}, \inlinetikz{-3pt}{0.3}{\exgadgetthree}, and \inlinetikz{-3pt}{0.3}{\exgadgetone}. Here the gray vertices were used for gluing, and the upper copy was joined with the lower copy everywhere else. Since in the connected case the instance consisted of the disjoint union of gadgets, in the ``dual'' case of disconnected forbidden graphs we join them. Thus, the constructed instances are join graphs of gadgets where each gadget is the join of two copies of $H$ glued at an arbitrary vertex. Just as in the proof by Chen et al. we can construct $\vert H \vert^{\opt}$ such instances which all need different advice strings in order to be handled optimally, and therefore the lower bound of $\opt \cdot \log \vert H \vert$ holds also for disconnected $\F$.

These constructions are interesting for various reasons. First of all, we will later see that for certain mixed families $\F$ that may contain connected and disconnected graphs simultaneously, a lower bound linear in $\opt$ may still hold by using exactly these constructions. In particular we examine under which circumstances the proof of \label{thm:lb_disconnected} yields a lower bound on the advice complexity linear in $\opt$. Secondly, we will see in the next subsection that similar constructions even work for the \HEDD{} and result in an essentially tight lower bound on its advice complexity.

\subsection{H-Edge Deletion Problem}
Chen et al. previously proved a lower bound of $\opt \cdot \log \| H \|$ on the advice complexity of the \HEDD{} for connected graphs $H$ which essentially matches the trivial upper bound from Theorem \ref{thm:trivial_ub_edge_deletion}. We show that the same bound even holds if $H$ is disconnected. For the node-deletion problem with disconnected $\F$ we constructed instances that make extensive use of the join operation. We will see that similar constructions can be used for the edge-deletion case. It will be insightful to understand why exactly the join operation behaves nicely with disconnected graphs. The most important observation is summarized in the following lemma.

\begin{lemma}\label{lemma:join_lemma}
    Let $H$ be a disconnected graph and let $G_1, G_2$ be two other
    graphs. Then $G_1 \join G_2$ is $H$-free iff $G_1$ and
    $G_2$ are $H$-free.
\end{lemma}
\begin{proof}
    Obviously, if $G_1 \join G_2$ is $H$-free the same must hold for
    all its induced subgraphs, in particular $G_1$ and $G_2$. For the 
    converse, we assume that $G_1$ and $G_2$ are $H$-free but $W 
    \subseteq V(G_1 \join G_2)$ induces an $H$ in $G_1 \join G_2$. Since 
    $G_1$ and $G_2$ are $H$-free, $W$ cannot contain only vertices 
    from one of the two graphs. Hence, there exists a $v_1 \in V(G_1) 
    \cap W$ and a $v_2 \in V(G_2) \cap W$. These two vertices are 
    connected by an edge in $G_1 \join G_2$. Then any other vertex 
    $w \in W$ is either connected to $v_1$ or $v_2$ by an edge. Thus,
    $( G_1 \join G_2)[W]$ must be connected. However, $H$ is a disconnected graph,
    therefore, $W$ cannot induce an $H$ in $G_1 \join G_2$.
    
\end{proof}

We introduce the notion of an $e$-extension of a graph $H$. Intuitively, an $e$-extension is a graph $U_H(e)$ that extends $H$ in such a way that the unique optimal way to make $U_H(e)$ $H$-free is to delete the edge $e$.

\begin{definition}
    For a disconnected graph $H$ and an edge $e \in E(H)$ we call a graph $U_H(e)$ an \emph{$e$-extension} of $H$ if it satisfies
    \begin{itemize}
        \item[(E.1)] $H \isub U_H(e)$,
        \item[(E.2)] $H \not\phisub U_H(e) - e$, and
        \item[(E.3)] $H \phisub U_H(e) - f$ for all $f \in E(U_H(e)) \setminus \{e\}$.
    \end{itemize}
    We call $H$ \emph{edge-extendable} if for every $e \in E(H)$ there is such an $e$-extension.
\end{definition}

It turns out that extendability is a sufficient condition for the desired lower bound to hold.

\begin{theorem}\label{thm:extendable_lb}
    Let $H$ be an edge-extendable graph. Then any optimal online algorithm for the \HEDD{} needs $\opt \cdot \log \| H \|$ bits of advice.
\end{theorem}
\begin{proof}
    Let $m \in \N$ be arbitrary. We construct a family of instances with optimal solution size $m$ such that any optimal online algorithm needs advice to distinguish these instances. Take $m$ disjoint copies $H^{(1)}, ..., H^{(m)}$ of $H$. We denote the vertices of $H^{(i)}$ by $v^{(i)}$. Furthermore, let $e_1,...,e_m$ be arbitrary edges such that $e_i \in E(H^{(i)})$. We construct the instance $G(e_1,...,e_m)$ in $m$ phases. In the $i$th phase we reveal $H^{(i)}$ and join it with the already revealed graph from previous phases. Then we extend $H^{(i)}$ to $U_{H^{(i)}}(e_i)$ and again join the newly added vertices with the already revealed graph from the previous phases. If $G(e_1,...,e_{i-1})$ is the graph after phase $i-1$, after phase $i$ we have revealed a graph isomorphic to $G(e_1,...,e_{i-1}) \join U_{H^{(i)}}(e_i)$. Thus, $G := G(e_1,...,e_m) \iso U_{H^{(1)}}(e_1) \join ... \join U_{H^{(m)}}(e_m)$. 
    We claim that $X := \{e_1,...,e_m\}$ is the unique optimal solution for the $H$-Edge Deletion problem on $G$. Deleting all $e_i$ from $G$ yields a graph isomorphic to $(U_{H^{(1)}}(e_1) - e_1) \join ... \join (U_{H^{(m)}}(e_m) - e_m)$. By definition of an e-extension, and Lemma \ref{lemma:join_lemma}~ this graph is $H$-free. Thus $X$ is indeed a solution. It is furthermore optimal because $G$ contains $m$ edge-disjoint copies of $H$. Finally, if in one of the $U_{H^{(i)}}(e_i)$ we delete any other edge than $e_i$, by definition we need to delete at least one more edge to make $U_{H^{(i)}}(e_i)$ $H$-free. Hence, $X$ is the unique optimal solution. \\
    We can construct $\| H \|^m$ such instances that pairwise only differ in the choice of the edges $e_1,...,e_m$. Any online algorithm needs advice to distinguish these instances, and therefore requires $m \cdot \log \| H \|$ bits to be optimal on all of them. Since $m = opt$, the claim is proven.
\end{proof}

Next, we prove constructively that each disconnected graph $H$ without isolated vertices is edge-extendable. Later we deal with the case that $H$ has isolated vertices.

\begin{lemma}\label{lem:no_isol_extendable}
    Let $H$ be a disconnected graph without isolated vertices. Then $H$ is edge-extendable.
\end{lemma}
\begin{proof}
    Let $xy = e \in E(H)$ be an arbitrary edge. We construct an $e$-extension $U_H(e)$ of $H$ as follows.
    Let $H'$ be a disjoint copy of $H$ with vertex set $V(H') = \set{v'}{v \in V(H)}$. Now we identify $x$ with $x'$ and $y$ with $y'$ and join $H$ with $H'$ everywhere else. We call the resulting graph $U$. An example can be seen in figure \ref{fig:ex_extension}. We verify (E.1)-(E.3) for $U$. (E.1) is trivially fulfilled. Furthermore, $U$ satisfies (E.3) since $U$ consists of two copies of $H$ that only share the edge $e$. Deleting any other edge $f$ will leave one of the two copies unchanged, and thus $H \phisub U - f$. Finally, we prove (E.2) by contradiction. Suppose that there is a set $W \subseteq V(U)$ that induces a graph isomorphic to $H$ in $U^* := U - e$. If $W$ contains neither $x$ nor $y$, we have $H \phisub U^* - \{x,y\} = (H - \{x,y\}) \join (H' - \{x',y'\})$. However, the latter is $H$-free by Lemma \ref{lemma:join_lemma}. Thus, $W$ must contain $x$ or $y$. W.l.o.g. we assume that $x \in W$. Since $H$ has no isolated vertices, $W$ contains a neighbor of $x$ in $U^*$. We have $N^{U^*}(x) = (N^{H}(x) \setminus \{y\}) \cup (N^{H'}(x') \setminus \{y'\})$, thus w.l.o.g. we may assume that there is a $v \in W \cap N^{H}(x) \setminus \{y\}$. However, there must also exist a $w' \in W \cap (V(H') \setminus \{x',y'\})$: Otherwise $W \subseteq V(H)$, yet after the edge deletion $H - e$ is $H$-free. This implies $vw' \in E(U^*)$, and every vertex except for $y$ is either connected to $v$ or $w'$. However, if $y \in W$, also some neighbor of $y$ must be in $W$, and therefore we conclude that $W$ must induce a connected graph. This contradicts the assumption, and we have proven that $U$ is indeed an $e$-extension of $H$.
    
\end{proof}

\begin{figure}
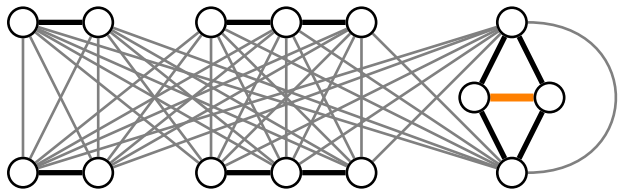

    \centering
    \exextension
    \caption{Example for the $e$-extension of the graph $P_2 \cup P_3 \cup K_3$ as constructed in Lemma \ref{lem:no_isol_extendable}. The edge $e$ is depicted in orange.}
    \label{fig:ex_extension}
\end{figure}

The results from Chen et al. together with Theorem \ref{lem:no_isol_extendable} and \ref{thm:extendable_lb} yield the following corollary.

\begin{corollary}\label{cor:heddnoisol}
    Let $H$ be a graph without isolated vertices. Then any optimal online algorithm for the \HEDD{} needs $\opt \cdot \log \| H \|$ bits of advice.
\end{corollary}

We now turn finally to the case where $H$ has isolated nodes. We prove via a simple advice-preserving reduction from the known case to the case where $H$ has isolated vertices that the same lower bound holds.

\begin{theorem}\label{thm:heddwithisol}
    Let $H$ be a graph with $k > 0$ isolated vertices. Then any optimal online algorithm for the \HEDD{} needs at least $\opt \cdot \log \| H \|$ bits of advice.
\end{theorem}
\begin{proof}
    Let $H'$ be the graph we obtain from $H$ by deleting all isolated vertices. Of course, $\| H \| = \| H' \|$. For any graph $G'$ we have: $G'$ is $H'$-free iff $\overline{K_k} \cup G'$ is $H$-free. Let $G'$ be an online instance for the \EDD{H'}. We construct an online instance $G$ for the \HEDD{} that presents $\overline{K_k}$ in the first $k$ time steps and then continues to present $G'$ node by node. Note that $\vert G \vert = \vert G' \vert + k$. The deletions in $G'$ can be translated in to deletions in $G$ by shifting $k$ time steps. The optimal solutions of $G$ and $G'$ coincide up to this shifting by $k$ time steps, and of course $\opt_H(G) = \opt_{H'}(G') = \opt$. Thus, the advice complexity for the \EDD{H} is at least the advice complexity for the \EDD{H'}. For this problem, however, we already have a lower bound from Corollary \ref{cor:heddnoisol}. Thus, the same lower bound applies to the case where $H$ has isolated vertices. 
    
\end{proof}

\section{The \FNDD}
We have seen that for obstruction sets, in which all or none of the
graphs are connected, the advice complexity is linear in the number of
optimal deletions. This is not always the case when considering
general families of graphs $\F$ as obstruction sets. An easy example is
the following: consider the family $\F_4$ that contains all graphs
over four nodes.  Clearly, whenever any
fourth node of an online graph is revealed, a node has to be deleted.
Yet which concrete node is deleted is arbitrary for an optimal
solution as every solution will have to delete all but three nodes of
the complete instance.  Thus, no advice is needed.

A more curious observation is that there are families of forbidden
graphs that need advice that is logarithmic in the order of the
optimal number of deletions. Further, logarithmic advice is also
sufficient to optimally solve most of these problems. This is due to
the fact that depending on the forbidden family of graphs we can bound
the number of \emph{remaining} nodes of an instance after an optimal number
of deletions has been made.

Finally, we classify some families that contain connected and
disconnected graphs that require linear advice by having a closer look
at which families of forbidden graphs are compatible with a join
construction.
\subsection{Logarithmic Advice}
We start by observing that when an $\F$ contains both an independent
set and a clique, the size of the biggest graph that contains no
$H \in \F$ is bounded.
\begin{lemma}\label{lemma:ramsey}
    Let $\F$ be an arbitrary family of graphs. Then there exists a
    minimal $R \in \N$ such that all graphs of size at least $R$ are
    not $\F$-free iff $K_n, \overline{K_m} \in \F$ for some $n, m \in
    \N$.
\end{lemma}
\begin{proof}
    Ramsey's Theorem guarantees the existence of $R$ if $K_n,
    \overline{K_m} \in \F$. Conversely, if $\F$ contains no clique
    (independent set), then arbitrarily big cliques (independent sets)
    are $\F$-free.
\end{proof}
We can use this observation to construct an algorithm that is mainly
concerned with the remaining graph after all deletions have been made.
As the size of this graph is bounded by a constant, we can simply tell
an online algorithm when it should \emph{not} delete every node of an $H$
that it sees and which one that is.
\begin{theorem}\label{thm:nodelogkub}
    If $K_n, \overline{K_m} \in \F$ for some $n,m \in \N$, then there
    is an optimal online algorithm with advice for the \FNDD{} that uses
    $O(\log \opt)$ bits of advice.
\end{theorem}
\begin{proof}
    Let $R$ be as in Lemma~\ref{lemma:ramsey}, and $k$ be the size of
    the biggest graph in $\F$. Algorithm~\ref{alg:logramsey} uses at
    most
    \begin{equation*}
        \lceil \log (R-1) \rceil + \lceil (R-1) \cdot \log (\opt \cdot k) \rceil = (R-1) \log (\opt) + O(1)
    \end{equation*}
    bits of advice.  We assume that the algorithm is given $\opt$
    beforehand, which can be encoded using $O(\log \opt)$ bits of
    advice using self-delimiting encoding as in~\cite{BKKK17}. The
    advisor computes an optimal offline solution. After deleting the
    nodes from $G$, a graph with at most $R-1$ nodes must remain,
    otherwise it would not be $\F$-free by Lemma~\ref{lemma:ramsey}.
    Let $u \leq R-1$ be the number of nodes that are considered by the
    online algorithm below and that will \emph{not} be deleted. The
    advisor writes $u$ onto the tape. Next, for all those $u$ nodes
    $(v_i)_{i \leq u}$, the advisor computes in which round the
    algorithm considers this node for the first time. 
    A node is
    considered by the algorithm if it is part of an induced forbidden
    graph $H \in \F$ that at some point is recognized. The node
    $(v_i)$ can thus be identified by a pair $(r_i, a_i) \in
    \{1,\ldots,\opt\} \times \{1,\ldots,k\}$. Then the algorithm encodes all
    these pairs $\left((r_i, a_i)\right)_{i \leq u}$ onto the tape.

    The algorithm starts by reading $u$ and these pairs from the tape.
    Then it sets its \emph{round counter} $r$ to $1$, and the set of
    \emph{fixed nodes} (i.e.\ the set of nodes that it encounters and
    which will not be deleted) $F$ to $\emptyset$. Whenever the
    algorithm finds a forbidden induced subgraph it checks in the list
    $(r_i, a_i)$ which of its nodes it must not delete, and adds them
    to $F$. Then it deletes any other vertex from $W \setminus F$.
    
\end{proof}
    \begin{algorithm}
        \caption{Optimal Online Algorithm with Logarithmic Advice}
        \begin{algorithmic}[1]
            \State Read $\lceil \log (R-1) \rceil$ bits of advice, interpret as number $u \in \{1,\ldots,R-1\}$
            \State Read $\lceil u \cdot \log (k \cdot \opt) \rceil$ bits of advice, interpret as $u$ pairs $\left((r_i, a_i)\right)_{i \leq u} \subseteq \{1,\ldots,\opt\} 
            \times \{1,\ldots,k\}$
            \State $r \gets 1, F \gets \emptyset$
            \ForAll{$t = 1,\ldots,T}$
                \State $G_t \gets $ reveal next node
                \While{$G_t[W] \iso H \in \F$ for some $W \subseteq V(G_t)$}
                    \ForAll{$i=1,\ldots,u$}
                        \If{$r == r_i$}
                            \State $v_i \gets a_i$'th vertex of $W$
                            \State $F \gets F \cup \{v_i\}$
                        \EndIf
                    \EndFor
                    \State Delete any vertex from $W \setminus F$ (or all at once)
                    \State $r \gets r + 1$
                \EndWhile
            \EndFor
        \end{algorithmic}
        \label{alg:logramsey}
    \end{algorithm}

This proof implies that, given some $\F$, if we can always bound the
size of the graph remaining after deleting an optimal number of nodes
by a constant, we can construct an algorithm that solves the \FNDD{}
with advice logarithmic in \opt.  Under certain conditions we also get
a lower bound logarithmic in $\opt$ as we will see in the following
two theorems.

\begin{theorem}\label{thm:nodelogklbuni}
    Let $K_n, \overline{K_m} \in \F$, and let $D$ be a graph that is
    $\F$-free, $\vert D \vert = R-1$, and $\| D \|$ is maximal among
    such graphs. If $D$ has no universal vertex, then any optimal online
    algorithm for the \FNDD{} needs $\Omega(\log \opt)$ bits of advice.
\end{theorem}
\begin{proof}
    Let $c$ be the size of the biggest clique in $D$. Of course $c <
    n$, otherwise $K_n \isub D$. Let $\opt > n$, and select $c$
    distinct numbers $u_1,\ldots,u_c \in \{1,\ldots,\opt+c\}$. Reveal
    $K_{\opt + c}$, then continue to add $\vert D \vert - c$ nodes and
    join them with all nodes of the $K_{\opt + c}$ except
    $u_1,\ldots,u_c$. Between these newly added vertices and
    $\{u_1,\ldots,u_c\}$ add edges such that they form a graph isomorphic
    to $D$. The entire resulting graph $G$ is then isomorphic
    to $K_{opt} \join D$. We claim that the unique optimal way to make
    $G$ $\F$-free is to delete $X := V(K_{\opt})$.
    
    First we observe that $X$ is a solution, because $G - X \iso D$.
    It is also optimal because if we delete fewer vertices then we are
    left with more than $\opt + \vert D \vert - \opt = \vert D \vert$
    nodes, and the resulting graph is not $\F$-free according to
    Lemma~\ref{lemma:ramsey}. For the uniqueness, suppose there was
    another optimal solution $X' \neq X$. Let $\overline{X},
    \overline{X'}$ denote the respective complement sets, i.e.\ the
    nodes that are not deleted in the respective optimal solution.
    Then because $D$ has no universal vertex we have
    \begin{equation*}
        \| G[\overline{X}] \| 
        < \| G[\overline{X} \cap \overline{X'}]\| + \vert \overline{X'} \setminus \overline{X} \vert \cdot (\vert D \vert - 1)
        = \| G[\overline{X'}] \|.
    \end{equation*}
    This would contradict the assumption that $D$ was chosen such that
    $\| D \|$ is maximized. Thus, we can construct $\binom{\opt +
    c}{c}$ such instances, each of which has a unique optimal
    solution, and the solutions are pairwise different. The nodes must
    be deleted already when the $K_{\opt + c}$ is revealed, hence we
    get a lower bound on the advice complexity of
    $\log \left( \frac{(\opt + c)\cdot \ldots \cdot (\opt +1)}{c!} \right) = \Omega(\log \opt).$
\end{proof}

With a similar construction for independent sets instead of cliques we
get the following sufficient condition for the necessity of
logarithmic advice.
\begin{theorem}\label{thm:nodelogklbisol}
    Let $K_n, \overline{K_m} \in \F$, and let $D$ be a graph that is
    $\F$-free, $\vert D \vert = R-1$, and $\| D \|$ is minimal among
    such graphs. If $D$ has no isolated vertex, then any optimal
    online algorithm for the \FNDD{} needs $\Omega(\log \opt)$ bits of
    advice.
\end{theorem}

\subsection{Linear Advice}
It is trivial to solve the \FNDD{}
using linear advice, as an algorithm can simply ask for each $H$ that
it encounters which node the optimal one to delete is. We have also seen
that for non-mixed $\F$ this amount of advice is also necessary.
We now show that for many mixed families $\F$, linear advice is
also necessary.

For this, we take a closer look at the graphs that remain after a glue
operation followed by a deletion at the glued node.  We characterize
the families of graphs for which this construction ensures that no
unintended copies of some $H \in \F$ are created by gluing. It
follows that all families $\F$ that are either \emph{sub-$H$-unions}
or \emph{sub-$H$-joins} as defined in the following require advice
linear in \opt.

Notice that we may w.l.o.g.\ assume that $H$ is connected, otherwise
consider the $\overline{\F}$-problem. However, for completeness we
cover both cases here.

\begin{definition}
    For some graph $H$, we call a graph $G$ a \emph{sub-$H$-union}  if
    all connected components of $G$ are proper induced subgraphs of
    $H$.
\end{definition}

%There are some broken links here referencing back to the chen et al. paper results
Notice that in the construction from Chen et al.~\cite{ChenHLR21} (Theorem 3)
for $H$, after all nodes have been deleted, any induced subgraph must
be a sub-$H$-union. Conversely, for every sub-$H$-union $G$, one of
the constructions in the work of Chen et al. would contain $G$ as an
induced subgraph after all nodes have been deleted. Hence, the
constructions work for $H$ iff $\F$
contains apart from $H$ only graphs that are \emph{not} sub-$H$-
unions. It immediately follows from the definition that a graph is
not a sub-$H$-union iff one of its connected components is not a
proper induced subgraph of $H$. Note in particular that $H$ is not a
sub-$H$-union iff $H$ is connected. Thus, the above construction only
makes sense if $H$ is connected.

\begin{lemma}
    $G$ is not a sub-$H$-union iff $G$ has an induced subgraph $U$ of
    size at most $\vert H \vert$ that is connected and not a proper
    induced subgraph of $H$.
\end{lemma}
\begin{proof}
    If $G$ is not a sub-$H$-union, then one of its connected
    components $C$ is not a proper induced subgraph of $H$. If this
    component has at least $\vert H \vert$ nodes, then it also contains a
    connected induced subgraph with $\vert H \vert$ nodes. This
    subgraph is then not a proper induced subgraph of $H$. If on the
    other hand $\vert C \vert < \vert H \vert$, then $C$ is not a
    proper induced subgraph of $H$ and ($\Rightarrow$) is proven.

    For ($\Leftarrow$), let $C$ be the connected component that
    contains $U$. Then $C$ cannot be a proper induced subgraph of $H$,
    thus $G$ is not a sub-$H$-union.
\end{proof}

This allows us to characterize the family of graphs that are not
sub-$H$-unions by a \emph{finite} family. We use the following
notation.

\begin{definition}
    For a set of graphs $\G$ we denote the set of all induced
    supergraphs of the graphs $G \in \G$ by $\G^{\uparrow}$.
\end{definition}
 
\begin{theorem}\label{thm:finite_family_not_subH_union}
    The set of all graphs that are not sub-$H$-unions is
    \begin{align*}
        &\set{G}{G \text{ connected}, \vert G \vert \leq \vert H \vert, G \text{ not a proper induced subgraph of } H}^{\uparrow} \\
    =& \left( \set{G}{G \text{ connected}, \vert G \vert < \vert H \vert, G \not\isub H}  \cup \set{G}{G \text{ connected}, \vert G \vert = \vert H \vert}\right)^{\uparrow}
    \end{align*}
\end{theorem}

\begin{definition}
    We call a graph $G$ a \emph{join graph} if we can partition
    $V(G)$ into two disjoint non-empty subsets $V_1, V_2$ such that $G
    = G[V_1] \join G[V_2]$. \\
    We call a partition of $V(G)$ into disjoint non-empty subsets
    $V_1,\ldots,V_k$ a \emph{join decomposition} of $G$ if the
    following properties are fulfilled:
    \begin{enumerate}
        \item $G = G[V_1] \join \ldots \join G[V_k]$.
        \item None of the $G[V_i]$ are join graphs.
    \end{enumerate}
    We call the graphs $G[V_i]$ \emph{join components} of $G$.
\end{definition}

\begin{lemma}\label{lem:join_decomp}
    Every graph has a unique join decomposition that can be computed in time $O(\vert G \vert^2)$.
\end{lemma}
\begin{proof}
    The graphs induced by vertex sets of the join components are exactly the connected components of the complement graph $\overline{G}$.
\end{proof}

\begin{definition}
    We call a graph $G$ a \emph{sub-$H$-join} if all join components of $G$ are proper induced subgraphs of $H$.
\end{definition}

Again we observe that the constructions from
Chen et al. leave behind exactly sub-$H$-joins. So,
these constructions work iff $\F$ contains apart from $H$ only graphs
that are not sub-$H$-joins, i.e.\, those graphs that have a join
component that is not a proper induced subgraph of $H$.  Disconnected
graphs~$H$ are not sub-$H$-joins of themselves. One might suppose that
analogously to the union construction, $H$ is a
not sub-$H$-join iff $H$ is disconnected. This is wrong, however,
because $P_4$ is not a join graph, and thus also not a join of proper
induced subgraphs of itself. Thus, the join construction not only
works for disconnected graphs~$H$, but also for some connected ones.

\begin{lemma}
    $G$ is a sub-$H$-join iff $\overline{G}$ is a sub-$\overline{H}$
    union.
\end{lemma}
\begin{proof}
    Let $\{C_i\}_{i \in I}$ be the connected components of
    $\overline{G}$. Then by Lemma~\ref{lem:join_decomp}, the components
    $\overline{C_i}$ are the join components of $G$. Now
    $\overline{C_i}$ is a proper induced subgraph of $H$ iff $C_i$ is
    a proper induced subgraph of $\overline{H}$. This proves the
    claim.
\end{proof}

Thus, the set of graphs that are not sub-$H$-joins can also be
characterized by a finite obstruction set according to
Theorem~\ref{thm:finite_family_not_subH_union}.

\begin{theorem}
    The set of graphs that are not sub-$H$-joins is
    \begin{align*}
        &\set{G}{\overline{G} \text{ connected}, \vert G \vert \leq \vert H \vert, \overline{G} \text{ not a proper induced subgraph of } \overline{H}}^{\uparrow} \\
        =& \left( \set{G}{\overline{G} \text{ connected}, \vert G \vert < \vert H \vert, G \not\isub H} \cup \set{G}{\overline{G} \text{ connected}, \vert G \vert = \vert H \vert}\right)^{\uparrow}
    \end{align*}
\end{theorem}

\section{Further Work}
While we were able to shed further light on the bigger picture of
online node and edge deletion problems with advice, the most general
problems of their kind are still not solved. For node deletion
problems, the case of an obstruction set with both connected and
disconnected graphs proves to be much more involved, with the advice
complexity being heavily dependent on the obstruction set, as we have
seen in the previous section.

The logarithmic bounds of this paper cannot be directly transferred to
the $\FEDD$, as independent sets cannot be part of the
obstruction set. There are, of course, families $\F$ for which no advice
is necessary, e.g., $\F = \{\inlinetikz{1pt}{0.3}{\exgraphtwo}\}$, but it seems hard to find
non-trivial families for which less than linear advice is both
necessary and sufficient. An additional difficulty is that forbidden graphs
may be proper (non-induced) subgraphs of one another, which makes it
difficult to count deletions towards individual copies of forbidden
graphs. Chen et al.~\cite{ChenHLR21} proposed a recursive way to do
so, but it is unclear if their analysis can be generalized to
arbitrary families of forbidden graphs $\F$.

\bibliography{fnodeedgedeletion}
\end{document}